\documentclass[11pt]{elsarticle}
\usepackage{amsmath,amssymb,amsthm}
\usepackage{hyperref}

\theoremstyle{plain}
\newtheorem{theorem}{Theorem}[section]
\newtheorem{lemma}[theorem]{Lemma}

\newtheorem{proposition}[theorem]{Proposition}
\newtheorem{corollary}[theorem]{Corollary}

\theoremstyle{definition}
\newtheorem{definition}[theorem]{Definition}

\newcommand{\dda}{\mathord{\mbox{\makebox[0pt][l]{\raisebox{-.4ex}{$\downarrow$}}$\downarrow$}}}

\newcommand{\ua}{\mathord{\uparrow}}
\newcommand{\da}{\mathord{\downarrow}}
\newcommand{\rom}[1]{\rm{\uppercase\expandafter{\romannumeral #1}}}

\newcommand{\bv}{\mathord{\bigvee}}

\begin{document}

\begin{frontmatter}

\title{A characterization of the consistent Hoare powerdomains over dcpos}

\author[1]{Zhongxi Zhang\corref{a1}}
\address[1]{School of Computer and Control Engineering, Yantai University, Yantai, Shandong, 264005, China}
\cortext[a1]{Corresponding author.}
\ead{zhangzhongxi89@gmail.com}

\author[2]{Qingguo Li}
\address[2]{College of Mathematics and Econometrics, Hunan University, Changsha, Hunan, 410082, China}
\ead{liqingguoli@aliyun.com}

\author[1] {Nan Zhang}
\ead{zhangnan0851@163.com}

\begin{abstract}
It has been shown that for a dcpo $P$, the Scott closure of $\Gamma_{c}(P)$ in $\Gamma(P)$ is a consistent Hoare powerdomain of $P$, where $\Gamma_{c}(P)$ is the family of nonempty, consistent and Scott closed subsets of $P$, and $\Gamma(P)$ is the collection of all nonempty Scott closed subsets of $P$. In this paper, by introducing the notion of a $\bv{-}$existing set, we present a direct characterization of the consistent Hoare powerdomain: the set of all $\bv{-}$existing Scott closed subsets of a dcpo $P$ is exactly the consistent Hoare powerdomain of $P$. We also introduce the concept of an $F$-Scott closed set over each dcpo-$\vee_{\ua}$-semilattice. We prove that the Scott closed set lattice of a dcpo $P$ is isomorphic to the family of all $F$-Scott closed sets of $P$'s consistent Hoare powerdomain.
\end{abstract}

\end{frontmatter}

\section{Introduction}

The Hoare powerdomain plays an important role in modeling the programming semantics of nondeterminism, which is analogous to the powerset construction (see for example \cite{HECKMANN199177,HECKMANN1992,HECKMANN2013215,Plotkin1976,SMYTH197823,SMYTH1983}). The Hoare power domain over a dcpo (directed complete poset) is the free inflationary semilattice where the inflationary operator is exactly the binary join operator. The standard construction of a Hoare powerdomain consists all nonempty Scott closed subsets $\Gamma(P)$ of the dcpo $P$, order given by the inclusion relation. In \cite{Yuan2014Consistent}, Yuan and Kou introduced a new type of powerdomain, called the consistent Hoare powerdomain. The new powerdomain over a dcpo is a free algebra where the inflationary operator delivers joins only for consistent pairs. They provided a concrete way of constructing the consistent Hoare powerdomain over a domain (continuous dcpo): the family of all nonempty relatively consistent Scott closed subsets of the domain is such a powerdomain. Follows from the work, there is a natural problem: whether every dcpo has a consistent Hoare powerdomain. Geng and Kou \cite{GENG2017169} gave an affirmative answer to the question by showing that the Scott closure $\overline{\Gamma_{c}(P)}$ of $\Gamma_{c}(P)$ in $\Gamma(P)$ is a consistent Hoare powerdomain over the dcpo $P$, where $\Gamma_{c}(P)$ is the collection of all nonempty consistent Scott closed subsets of $P$. One question arises naturally: can we give a direct characterization of the consistent Hoare powerdomain over every dcpo $P$? That is to say, what type of Scott closed subsets of $P$ is exactly the consistent Hoare powerdomain? This paper is mainly set to answer this question.

The consistent Hoare powerdomain can be viewed as a type of completion that embeds every dcpo into a dcpo which is also complete with respect to consistent pair joins. In \cite{zd}, Zhao and Fan introduced a type of dcpo-completion embedding each poset into a dcpo, called \mbox{$D$-completion}. They proved that the Scott closed set lattice of a poset is order isomorphic to that of the \mbox{$D$-completion}. It is also the case for the sobrification, i.e., the topology lattice $\mathcal{O}(X)$ of a $T_{0}$-space $X$ is  isomorphic to the topology lattice $\mathcal{O}(X^{s})$ of the sobrification $X^{s}$ of $X$ (see \cite{gg1}). In this paper, we introduce a type of closed sets on every dcpo-$\vee_{\ua}$-semilattice, called $F$-Scott closed sets. We prove that the Scott closed set lattice of a dcpo is isomorphic to the family of all $F$-Scott closed subsets of its consistent Hoare powerdomain. It is known that two sober dcpos are isomorphic if and only if their Scott set lattices are isomorphic (see \cite{HZP,Zhao2018U}). Consequently, the consistent Hoare powerdomains for sober dcpos are uniquely determined up to isomorphism.

\section{Preliminaries}

This section is an introduction to the concepts and results about the consistent Hoare powerdomains over domains and dcpos. For more details, refer to \cite{GENG2017169,gg1,Yuan2014Consistent}. In this paper, the order relation on each family of sets is the set-theoretic inclusion relation.

A nonempty subset $D$ of a poset $P$ is called \emph{directed} if every pair has an upper bound in~$D$. If every directed $D \subseteq P$ has a supremum $\bv D$, then $P$ is called a \emph{dcpo}. For $A \subseteq P$, let $\da A = \{x \in P : x \leq a$ for some $a \in A\}$, $\da x = \da \{x\}$ where $x \in P$. If $A = \da A$, then $A$ is called a \emph{lower set}. For $x, y \in P$, we say $x$ is \emph{way below} $y$, written $x \ll y$, if $y \leq \bv D$ implies $x \in \da D$ for any directed $D \subseteq P$. 
Let $\dda x = \{y \in P : y \ll x \}$, and $\dda A = \bigcup \{\dda a : a \in A\}$. A dcpo $P$ is called a \emph{domain} if each element $x$ is the directed join of $\dda x$. A subset $A$ of a dcpo $P$ is \emph{Scott closed} if it is a lower set and closed under directed sups. Let $cl(A)$ denote the Scott closure of $A$.

\begin{definition}\cite{Yuan2014Consistent}
A \emph{consistent inflationary semilattice} is a dcpo $P$ with a Scott continuous binary partial operator $\biguplus_{\ua}$ defined only for consistent pairs of points that satisfies three equations for commutativity $x \biguplus_{\ua} y = y \biguplus_{\ua} x$, associativity $x \biguplus_{\ua} (y \biguplus_{\ua} z) = (x \biguplus_{\ua} y) \biguplus_{\ua} z$ and idempotency $x \biguplus_{\ua} x = x$ together with the inequality $x \leq x \biguplus_{\ua} y$ for $x, y, z \in P$. The free consistent inflationary semilattice over a dcpo $P$ is called the \emph{consistent Hoare powerdomain} over $P$ and denoted by~$H_{c}(P)$.
\end{definition}

For a consistent inflationary semilattice, the operator $\biguplus_{\ua}$ coincides with $\vee_{\ua}$, the join operator defined only for consistent pairs. That is to say, a consistent inflationary semilattice is exactly a dcpo that is a $\vee_{\ua}$-semilattice, also called \emph{a dcpo-$\vee_{\ua}$-semilattice}.
A Scott closed subset $A$ of a domain $L$ is called \emph{relatively consistent} if $A$ is the Scott closure of the directed union $\bigcup^{\ua}\{\da F : F \in \mathcal{F}_{C}(A)\}$, where
\begin{center}
	$\mathcal{F}_{C}(A) = \{F : F $ is a nonempty finite consistent subset of $L$ and $F \subseteq \dda A\}$.
\end{center}
The set of all nonempty relatively consistent Scott closed subsets of $L$ is denoted by $R\Gamma_{C}(L)$.

\begin{theorem}\emph{\cite{Yuan2014Consistent}}
	Let $L$ be a domain. The consistent Hoare powerdomain $H_{c}(P)$ over $L$ can be realized as $R\Gamma_{C}(L)$ where $\biguplus_{\ua} = \bigcup$. The embedding $j$ of $L$ into $R\Gamma_{C}(L)$ is given by $j(x) = \da x$ for $x \in L$.
\end{theorem}

Let $\Gamma_{c}(P)$ denote the family of all nonempty, consistent and Scott closed subsets of a dcpo $P$. And let $\overline{\Gamma_{c}(P)}$ be the Scott closure of $\Gamma_{c}(P)$ in $\Gamma(P)$. The family $\overline{\Gamma_{c}(P)}$ was introduced by Geng and Kou \cite{GENG2017169}, in order to declare that the consistent Hoare powerdomain over every dcpo exists.

\begin{lemma}\label{gj} \emph{\cite{GENG2017169}}
	Let $P$ be a dcpo and $L$ is a dcpo-$\vee_{\ua}$-semilattice. If $f : P \rightarrow L$ is a Scott continuous function, then for any $A \in \overline{\Gamma_{c}(P)}$, $\bv f(A)$ exists in $L$.
\end{lemma}

\begin{theorem}\emph{\cite{GENG2017169}}\label{gjt}
	Let $P$ be a dcpo. Then the consistent Hoare powerdomain $H_{c}(P)$ over $P$ can be realized as the $\overline{\Gamma_{c}(P)}$, where $A \vee_{\ua} B = A \bigcup B$ whenever $A, B$ are consistent in $\overline{\Gamma_{c}(P)}$. The embedding $j$ of $P$ into $\overline{\Gamma_{c}(P)}$ is given by $j(x) = \da x$ for any $x \in P$.
\end{theorem}

\section{A direct characterization of the consistent Hoare powerdomain}

\begin{definition}\label{d1}
	 A subset $A$ of a dcpo $P$ is called $\bv{-}$existing if for any continuous function $f : P \rightarrow L$ mapping into a dcpo{-}$\vee_{\ua}$-semilattice $L$, $\bv f(A)$ always exists in $L$. 
\end{definition}

For any directed subset $D$ of $P$, $f(D)$ is directed in $L$ since continuous functions are monotone, and then $\bv f(D)$ exists in the dcpo-$\vee_{\ua}$-semilattice $L$. Hence every directed subset is $\bv{-}$existing. Similarly, every consistent finite subset is also $\bv${-}existing. Notice that an empty set is not $\bv{-}$existing because a dcpo{-}$\vee_{\ua}$-semilattice $L$ may not have a least element.

\begin{proposition}\label{p2}
	Let $f: P \rightarrow L$ be a continuous function from dcpo $L$ to dcpo-$\vee_{\ua}$-semilattice~$L$, and $A \subseteq P$. Then 
	
	\emph{(1)} $\bv f(A)$ exists $\Leftrightarrow$ $\bv f(cl(A))$ exists $\Rightarrow$ $\bv f(A) = \bv f(cl(A))$.
	
	\emph{(2)} $A$ is $\bv${-}existing if and only if $cl(A)$ is $\bv${-}existing.
\end{proposition}

\begin{proof}
	(1) Assume that $x$ is an upper bound of $f(A)$. Then $f(A) \subseteq \da x$ and $A \subseteq f^{-1}(\da x) \in \Gamma(P)$. Hence $cl(A) \subseteq f^{-1}(\da x)$, and then $f(cl(A)) \subseteq  \da x$, i.e., $x$ is an upper bound of $f(cl(A))$. The converse is clearly true. Thus $\bv f(A)$ exists iff $\bv f(cl(A))$ exists, and both imply that $\bv f(A) = \bv f(cl(A))$.
	
	(2) It is straightforward from (1) and Definition \ref{d1}.
\end{proof}

	For any dcpo $P$, we write $P^{\vee} =\{A \subseteq P : A$ is $\bv${-}existing and Scott closed$\}$. By Lemma \ref{gj}, we immediately have that $\overline{\Gamma_{c}(P)} \subseteq P^{\vee}$. We shall show that the converse inclusion is also true, and then a characterization of the consistent Hoare powerdomain is obtained.

\begin{definition}
	A subset $A$ of a dcpo-$\vee_{\ua}$-semilattice $L$ is called \emph{$F$-Scott closed} if it is Scott closed and for any consistent nonempty finite set $F \subseteq A$, $\bv F \in A$. We write $\Gamma_{F}(L)$, called the \emph{$F$-Scott closure system} on $L$, for the set of all $F$-Scott closed subsets of $L$, and let $cl_{F}$ denote the corresponding closure operator. A function $f : L \rightarrow M$ between dcpo-$\vee_{\ua}$-semilattices is called \emph{$F$-Scott continuous} if $f$ is continuous with respect to the $F$-Scott closure systems.
\end{definition}

\begin{proposition}\label{p3.4}
	\emph{(1)} A function $f : L \rightarrow M$ between dcpo-$\vee_{\ua}$-semilattices is a dcpo{-}$\vee_{\ua}$-semilattice homomorphism iff $f$ is $F$-Scott continuous.
	
	\emph{(2)} If a nonempty subset $A$ of a dcpo-$\vee_{\ua}$-semilattice $L$ is consistent, then $cl_{F}(A) = \da \bv A$.
\end{proposition}

\begin{proof}
	(1) Suppose that $f$ is a dcpo{-}$\vee_{\ua}$-semilattice homomorphism. Let $A \subseteq M$ be $F$-Scott closed. We shall show that $f^{-1}(A) \subseteq L$ is also $F$-Scott closed. For any consistent nonempty finite $F \subseteq f^{-1}(A)$, we have $f(F)$ is also nonempty finite and consistent since $f$ is monotone. Then $f(\bv F) = \bv f(F) \in A$ and then $\bv F$ is in $f^{-1}(A)$. Similarly, $f^{-1}(A)$ is also closed with respect to directed joins. Thus $f$ is $F$-Scott continuous.
	
	Conversely, let $f$ be $F$-Scott continuous. Clearly, $f$ is monotone. Let $F$ be a consistent nonempty finite subset of $L$. Then $\bv f(F) \leq f(\bv F)$. The set $\da \bv f(F)$ is $F$-Scott closed in $M$. Then $f^{-1}(\da \bv f(F))$ is also $F$-Scott closed and $F \subseteq f^{-1}(\da \bv f(F))$. Hence $\bv F \in f^{-1}(\da \bv f(F))$ and then $f(\bv F) \leq \bv f(F)$. Thus $f(\bv F) = \bv f(F)$. Analogously, $f(\bv D) = \bv f(D)$ for all directed $D \subseteq L$. Therefore, $f$ is a dcpo{-}$\vee_{\ua}$-semilattice homomorphism.
	
	(2) Let $C \subseteq L$ be any $F$-Scott closed set with $A \subseteq C$. Notice that each subset of $A$ is also consistent. Then $D := \{\bv F : F \subseteq A$ is nonempty and finite$\}$ is a directed subset of $C$, and hence $\bv D = \bv A \in C$. Moreover, every $F$-Scott closed set is a lower set. Then $\da \bv A \subseteq C$, and thus $cl_{F}(A) = \da \bv A$.
\end{proof}

\begin{definition}
	Let $L$ be a dcpo-$\vee_{\ua}$-semilattice. A subset $A \subseteq L$ is called \emph{$F$-$\bv$-existing} if for any dcpo{-}$\vee_{\ua}$-semilattice homomorphism $f : L \rightarrow M$ mapping into a dcpo-$\vee_{\ua}$-semilattice $M$, $\bv f(A)$ exists in $M$. 
\end{definition}

\begin{lemma}\label{l3.6}
	Let $L$ be a dcpo-$\vee_{\ua}$-semilattice. A subset $A \subseteq L$ is $F$-$\bv$-existing iff $cl_{F}(A)$ is $F$-$\bv$-existing.
\end{lemma}

\begin{proof}
	The process is similar to that of Proposition \ref{p2}. Notice that every principal ideal $\da x$ is $F$-Scott closed.
\end{proof}

\begin{lemma}\label{l3.7}
	If a subset $A$ of a dcpo-$\vee_{\ua}$-semilattice $L$ is $F$-Scott closed and $F$-$\bv$-existing, then $A = \da a$ for some $a \in L$.
\end{lemma}

\begin{proof}
	The identity function $id : L \rightarrow L$ is a dcpo{-}$\vee_{\ua}$-semilattice homomorphism. We have that $\bv id(A) = \bv A$ exists in $L$. Then every nonempty finite subset $F \subseteq A$ is consistent, and hence $\bv F$ exists. Since $A$ is $F$-Scott closed, we obtain that $D = \{\bv F : F \subseteq  A$ is nonempty and finite$\}$ is a directed subset of $A$, and then $\bv D = \bv A \in A$. Thus $A = \da \bv A$, which completes the proof.
\end{proof}

\begin{lemma}\label{l3.8}
	A subset $A$ of a dcpo $P$ is $\bv$-existing iff $j(A)$ is an $F$-$\bv$-existing subset of $\overline{\Gamma_{c}(P)}$. \emph{(}Notice that $j$ is the function in Theorem \emph{\ref{gjt}.)}
\end{lemma}

\begin{proof}
Suppose that $A \subseteq P$ is $\bv$-existing. Let $f : \overline{\Gamma_{c}(P)} \rightarrow L$ be any dcpo{-}$\vee_{\ua}$-semilattice homomorphism mapping into a dcpo{-}$\vee_{\ua}$-semilattice $L$. Then $f \circ j$ is a Scott continuous function from $P$ to $L$, and hence $\bv f \circ j (A) = \bv f(j(A))$ exists in $L$. Thus $j(A)$ is an $F$-$\bv$-existing subset of $\overline{\Gamma_{c}(P)}$.

Now assume that $j(A)$ is an $F$-$\bv$-existing subset of $\overline{\Gamma_{c}(P)}$. Let $f : P \rightarrow L$ be a Scott continuous mapping into a dcpo{-}$\vee_{\ua}$-semilattice $L$. Then by Theorem \ref{gjt}, there exists a unique dcpo{-}$\vee_{\ua}$-semilattice homomorphism $\tilde{f} : \overline{\Gamma_{c}(P)} \rightarrow L$ such that $f = \widetilde{f} \circ j$. Then $\bv \tilde{f}(j(A)) = \bv \widetilde{f} \circ j(A) = \bv f(A)$ exists. Thus $A$ is a  $\bv$-existing subset of $P$.
\end{proof}

We now come to the characterization of the consistent Hoare powerdomains by $\bv$-existing Scott closed sets:

\begin{theorem}
	Let $P$ be a dcpo and $A \subseteq P$. Then $A \in \overline{\Gamma_{c}(P)}$ iff $A$ is Scott closed and $\bv$-existing, i.e., $\overline{\Gamma_{c}(P)} = P^{\vee}$.
\end{theorem}

\begin{proof}
	By Lemma \ref{gj}, we have that $\overline{\Gamma_{c}(P)} \subseteq P^{\vee}$. Now suppose that $A \in P^{\vee}$. Then, by Lemma~\ref{l3.8}, the set $j(A)$ is an $F$-$\bv$-existing subset of $\overline{\Gamma_{c}(P)}$. And by Lemma \ref{l3.6}, the \mbox{$F$-Scott} closed set $cl_{F}(j(A))$ is also $F$-$\bv$-existing. Then $cl_{F}(j(A))$ has a supremum which is in itself by Lemma \ref{l3.7}. Since the order on $\overline{\Gamma_{c}(P)}$ is the inclusion relation, we have $\bv cl_{F}(j(A)) = \bigcup cl_{F}(j(A)) \in \overline{\Gamma_{c}(P)}$. For each $a \in A$, $j(a) = \da a \in cl_{F}(j(A))$, and then $A \subseteq \bigcup cl_{F}(j(A))$. Thus $A \in \overline{\Gamma_{c}(P)}$ because $\overline{\Gamma_{c}(P)}$ is a lower set of $\Gamma(P)$, which completes the proof.
\end{proof}

\begin{theorem}\label{t3.10}
	Let $P$ be a dcpo. The family of Scott closed sets $\Gamma_{0}(P) = \Gamma(P) \bigcup \{\emptyset\}$ is order isomorphic to the $F$-Scott closure system $\Gamma_{F}(H_{c}(P))$ on the consistent Hoare powerdomain $H_{c}(P)$, i.e., $\Gamma_{0}(P) \cong \Gamma_{F}(H_{c}(P))$.
\end{theorem}

\begin{proof}
	Define a function $\eta : \Gamma_{0}(P) \rightarrow \overline{\Gamma_{c}(P)}$ by
	\begin{center}
		$\eta(A) = cl_{F}(j(A))$,
	\end{center}
	where $cl_{F}(j(A))$ is the $F$-Scott closure of $j(A) = \{\da a : a \in A\}$ in $\overline{\Gamma_{c}(P)}$. We shall show that $\eta$ is an order isomorphism. Obviously, $\eta$ is monotone. 
	
	Firstly, we prove that $\eta$ is injective. Suppose that $A, B \in \Gamma_{0}(P)$ with $A \neq B$. Without loss of generality, we may assume that there is $b \in B {\setminus} A$. Let $\mathcal{C} = \{C \in \overline{\Gamma_{c}(P)} : C \subseteq A\}$. Clearly, $j(A) \subseteq \mathcal{C}$ and $\da b \notin \mathcal{C}$. We claim that $\mathcal{C}$ is $F$-Scott closed in $\overline{\Gamma_{c}(P)}$. Then $cl_{F}(j(A)) \subseteq \mathcal{C}$, and hence $\da b \notin cl_{F}(j(A))$. But $\da b \in j(B) \subseteq cl_{F}(j(B))$. Thus $\eta(A) \neq \eta(B)$, i.e., $\eta$ is injective. Indeed, let $\mathcal{D} \subseteq \mathcal{C}$ be directed, then $\bv \mathcal{D}$, the supremum of $\mathcal{D}$ in $\overline{\Gamma_{c}(P)}$, is the Scott closure of $\bigcup \mathcal{D}$ in $P$, and hence $\bv \mathcal{D} \subseteq A$ since each element of $\mathcal{D}$ is contained in $A$, i.e., $\bv \mathcal{D} \in \mathcal{C}$. And if $C_{1}, C_{2} \in \mathcal{C}$, then $C_{1} \bigcup C_{2} \subseteq A$, thus $C_{1} \vee C_{2} \in \mathcal{C}$ (notice that $\overline{\Gamma_{c}(P)}$ is closed under finite unions), which proves the claim. 
	
	We next prove that $\eta$ is surjective. Suppose that $\mathcal{A}$ is $F$-Scott closed in $\overline{\Gamma_{c}(P)}$. We claim that $\bigcup \mathcal{A}$ is Scott closed in $P$. Indeed, if $D \subseteq \bigcup \mathcal{A}$ is directed, then $j(D) \subseteq \mathcal{A}$ is directed, and hence $\bv j(D) = \da \bv D \in \mathcal{A}$, i.e., $\bv D \in \bigcup \mathcal{A}$.
	We shall show that $\eta(\bigcup \mathcal{A}) = \mathcal{A}$. Since $j(\bigcup \mathcal{A}) \subseteq \mathcal{A}$, we have $\eta(\bigcup \mathcal{A}) = cl_{F}(j(\bigcup \mathcal{A})) \subseteq \mathcal{A}$. Conversely, assume that $A \in  \mathcal{A}$. Then $j(A)$ is consistent in~$\overline{\Gamma_{c}(P)}$. By Proposition \ref{p3.4}(2), we have $\bv j(A) = A \in cl_{F}(j(A)) \subseteq cl_{F}(j(\bigcup \mathcal{A}))$. Thus $\eta(\bigcup \mathcal{A}) = cl_{F}(j(\bigcup \mathcal{A})) = \mathcal{A}$, i.e., $\eta$ is surjective, which completes the proof.	
\end{proof}

A dcpo is called \emph{sober} if every Scott closed irreducible set is of the form $\da x$. It is known that for sober dcpos $P_{1}$ and $P_{2}$, $\Gamma_{0}(P_{1}) \cong \Gamma_{0}(P_{2})$ iff $P_{1} \cong P_{2}$ (see \cite{HZP,Zhao2018U} for more results about uniqueness of dcpos based on the Scott closed set lattices). In particularly, the Scott topology of every quasicontinuous domain is sober (each domain is a quasicontinuous domain, see \cite{gg1} for the detailed definition of a quasicontinuous domain). By applying the above theorem, we immediately have that for sober dcpos, the consistent Hoare powerdomains are uniquely determined up to isomorphism:

\begin{corollary}
	If dcpos $P_{1}$ and $P_{2}$ are sober, then $H_{c}(P_{1}) \cong H_{c}(P_{2})$ iff $P_{1} \cong P_{2}$.
\end{corollary}

\begin{proof}
	By Theorem \ref{t3.10}, $H_{c}(P_{1}) \cong H_{c}(P_{2})$ implies $\Gamma_{F}(H_{c}(P_{1})) \cong \Gamma_{F}(H_{c}(P_{2}))$ iff $\Gamma_{0}(P_{1}) \cong \Gamma_{0}(P_{2})$ iff $P_{1} \cong P_{2}$ implies $H_{c}(P_{1}) \cong H_{c}(P_{2})$.
\end{proof}

\section{Acknowledgments}
This work is supported by National Natural Science Foundation of China (No.11801491,11771134) and Shandong Provincial Natural Science Foundation, China (No. ZR2018BA004).



\end{document}